\documentclass[aps,prl,twocolumn,superscriptaddress,floatfix,nofootinbib,showpacs,longbibliography]{revtex4-1}

\usepackage[utf8]{inputenc}  
\usepackage[T1]{fontenc}     
\usepackage[table]{xcolor}
\usepackage[british]{babel}  
\usepackage[sc,osf]{mathpazo}
\usepackage[scaled=0.86]{berasans}  
\usepackage[colorlinks=true, citecolor=blue, urlcolor=blue]{hyperref}  
\usepackage{graphicx} 
\usepackage[babel]{microtype}  
\usepackage{amsmath,amssymb,amsthm,bm,amsfonts,mathrsfs,bbm} 

\usepackage{xspace}  
\usepackage{pgf,tikz}
\usepackage{xcolor}
\usepackage{multirow}
\usepackage{array}
\usepackage{bigstrut}
\usepackage{braket}
\usepackage{color}
\usepackage{natbib}
\usepackage{multirow}
\usepackage{mathtools}
\usepackage{float}
\usepackage[caption = false]{subfig}
\usepackage{xcolor,colortbl}
\usepackage{color}
\usepackage{rotating}
\usepackage{tikz}
\usetikzlibrary{quantikz}

\newcommand{\Tr}{\operatorname{Tr}}

\newcommand{\be}{\begin{equation}}
\newcommand{\ee}{\end{equation}}
\newcommand{\ba}{\begin{eqnarray}}
\newcommand{\ea}{\end{eqnarray}}

\newcommand{\tr}{\operatorname{Tr}}

\newtheorem{theorem}{Theorem}
\newtheorem{corollary}{Corollary}
\newtheorem{definition}{Definition}
\newtheorem{proposition}{Proposition}
\newtheorem{observation}{Observation}

\newtheorem{lemma}{Lemma}

\def\>{\rangle}
\def\<{\langle}

\usepackage{centernot}
\usepackage{subfig}

\begin{document}

\title{Composition of multipartite quantum systems: perspective from time-like paradigm}

\author{Sahil Gopalkrishna Naik}
\affiliation{School of Physics, IISER Thiruvananthapuram, Vithura, Kerala 695551, India.}

\author{Edwin Peter Lobo}
\affiliation{School of Physics, IISER Thiruvananthapuram, Vithura, Kerala 695551, India.}

\author{Samrat Sen}
\affiliation{School of Physics, IISER Thiruvananthapuram, Vithura, Kerala 695551, India.}

\author{Ram Krishna Patra}
\affiliation{School of Physics, IISER Thiruvananthapuram, Vithura, Kerala 695551, India.}

\author{Mir Alimuddin}
\affiliation{School of Physics, IISER Thiruvananthapuram, Vithura, Kerala 695551, India.}

\author{Tamal Guha}
\affiliation{Department of Computer Science, The University of Hong Kong, Pokfulam Road 999077, Hong Kong.}

\author{Some Sankar Bhattacharya}
\affiliation{International Centre for Theory of Quantum Technologies, University of Gdansk, Wita Stwosza 63, 80-308 Gdansk, Poland.}

\author{Manik Banik}
\affiliation{School of Physics, IISER Thiruvananthapuram, Vithura, Kerala 695551, India.}

\begin{abstract}
Figuring out the physical rationale behind natural selection of quantum theory is one of the most acclaimed quests in quantum foundational research. This pursuit has inspired several axiomatic initiatives to derive mathematical formulation of the theory by identifying general structure of state and effect space of individual systems as well as specifying their composition rules. This generic framework can allow several consistent composition rules for a multipartite system even when state and effect cones of individual subsystems are assumed to be quantum. Nevertheless, for any bipartite system, none of these compositions allows beyond quantum space-like correlations. In this letter we show that such bipartite compositions can admit stronger than quantum correlations in the time-like domain and, hence, indicates pragmatically distinct roles carried out by state and effect cones. We discuss consequences of such correlations in a communication task, which accordingly opens up a possibility of testing the actual composition between elementary quanta.
\end{abstract}



\maketitle
{\it Introduction.--} The idea of composition plays a crucial role in fabricating our worldview and accordingly guides us while constructing theories for the physical world \cite{Hardy13}. For instance, the objects we encounter in our daily life -- household items, machines, communication devices, computers {\it etc} -- can be thought of as being composed of more elementary parts, whereas {\it Penrose tribar} describes an interesting composite object yet impossible to exist \cite{Penrose58}. On the other hand, a region of spacetime with fields on it can be thought of as being composed of many smaller regions joined at their boundaries \cite{HardyGR}. Quantum formalism also assumes a particular composition rule while describing systems consisting of more than one subsystem \cite{Weyl31,Einstei35,Ozaw04,Zurek09,Masanes19}. Considering the individual systems to be quantum, one can construct several mathematically consistent models to describe state and effect spaces of a multipartite system, where consistency demands the outcome probability obtained from any pair of valid state and effect to be a positive number between zero and one. Constraints arising from physical and/or information theoretic demand may further abridge the scope of possible compositions. For instance, the state space of a bipartite system satisfying {\it no signaling} principle and {\it tomographic locality} postulate lies within two extremes -- {\it minimal} tensor product and {\it maximal} tensor product \cite{Hardy11,Namioka69,Barker,Auburn20}. The corresponding effect spaces are specified in accordance with the `no-restriction' hypothesis \cite{Chiribella10} which demands any mathematically well-defined measurement to be physically allowed. For brevity, the resulting theories arising from these two compositions will be denoted by $SEP$ and $\overline{SEP}$, respectively. In between these two extremes, many other compositions can be introduced among which quantum ($\mathcal{Q}$) composition is one example. 

Naturally the question arises whether there exist input-output correlations that are specific to some particular composite structure. An interesting answer stems from the Bell nonlocal correlations \cite{Bell64,Bell66,Mermin93,Brunner14} that are unavailable in $SEP$ theory, whereas all other compositions contain such nonlocal correlations. On the other hand, a {\it no-go} result can be argued from the work of Barnum {\it et al.} \cite{Barnum10} -- any bipartite space-like correlation obtained in $\overline{SEP}$ is also achievable in $\mathcal{Q}$. In fact, any composite model of two quantum systems satisfying the no signaling principle cannot have a beyond quantum space-like separated correlation within its description. It might be tempting to presume that any input-output correlation obtained in $SEP$ should also be achievable in $\mathcal{Q}$, since the role of state and effect cones in $SEP$ and $\overline{SEP}$ theories just get interchanged (see Fig.\ref{fig1}). In this letter we show that this naive intuition is, in-fact, not correct. There, indeed, exist time-like correlations in $SEP$ that cannot be obtained in $\mathcal{Q}$ composition. In fact, such correlations, with different strength, might exist in different composite models indicating empirical distinction among these compositions. In other words, mathematically consistent bipartite composition of elementary systems can have beyond quantum time-like correlations even when the elementary systems allow quantum description. We establish the above thesis through a communication task (game) involving two parties. We first analyze the optimal qubit communication required to accomplish the task when quantum composition is assumed. In fact, a necessary and sufficient condition for perfect accomplishment of the task is derived in generalized probabilistic theory (GPT) framework. We then show that for perfect success of the game the required number of qubits to be communicated is strictly less than quantum if $SEP$ composition is considered. Thus $SEP$ composition makes certain communication complexity problems trivial in comparison to quantum composition. It is worth mentioning that starting with an information-theoretic axiom  `{\it communication complexity is nontrivial}' researchers have identified `unphysical' consequences of beyond quantum space-like correlations \cite{Popescu94,vanDam13,Brassard06,Forster09,Brunner09,Shutty20} (see \cite{Buhrman10} for a review on communication complexity). Our study establishes that the same axiom can be efficiently utilized to isolate beyond quantum correlations in time-like domain.
Our proposed communication task provides an empirically testable criterion towards natural selection of the bipartite composite structure among different possible compositions lying in between $SEP$ and $\overline{SEP}$. 

{\it Preliminaries.--} First we briefly recall the framework of GPT since some of our results will be proved in this generalized framework. Although the framework had originated much earlier \cite{Mackey63,Ludwig67,Davies70}, advent of quantum information theory brought renewed interest in this framework \cite{Hardy01,Barrett07,Chiribella10,Barnum11,Masanes11}. For an overview of this framework we also refer to the works \cite{Gudder79,Busch97,Barnum06,Banik15,Plavala21}. A GPT is specified by a list of system types and the composition rules specifying combination of several systems. In prepare and measure scenario, a system ($\mathcal{S}$) is described by the tuple $(\Omega,E)$. Here $\Omega$ represents the collection of normalized states, generally a compact-convex set embedded in some positive cone $V_+$ (collection of unnormalized states) of some real vector space $V$. On the other hand, $E$ denotes the collection of effects, where an effect $e\in E$ corresponds to a linear functional $e:\Omega\mapsto[0,1]$, with $e(\omega)$ denoting the success probability of filtering the effect $e$ on the state $\omega\in\Omega$ when some measurement $M:=\{e_i~|~e_i\in E~\forall~i,~\sum_ie_i=u,~\&~u(\omega)=1~\forall~\omega\in\Omega\}$ is performed. The unnormalized effects form a cone $V^\star_+$ which is dual to the state cone $V_+$. For instance, state cone of a quantum system associated with Hilbert space $\mathcal{H}$ consists of the set of positive operators $\mathcal{T}_+(\mathcal{H})\subset\mathcal{T}(\mathcal{H})$ acting on $\mathcal{H}$, where $\mathcal{T}(\mathcal{H})$ denotes the set of all Hermitian operators on $\mathcal{H}$. A normalized state $\rho$ is an element of $\mathcal{T}_+(\mathcal{H})$ with trace one and their collection is the convex-compact set of density operators $\mathcal{D}(\mathcal{H})$. A generic quantum measurement corresponds to positive operator valued measure (POVM) $M:=\{\pi_i~|~\pi_i\in \mathcal{T}_+(\mathcal{H}),~\sum_i\pi_i=\mathbf{1}_{\mathcal{H}}\}$, with $\mathbf{1}_{\mathcal{H}}$ being the identity operator on $\mathcal{H}$.
\begin{figure}[t!]
\centering
\includegraphics[width=0.5\textwidth]{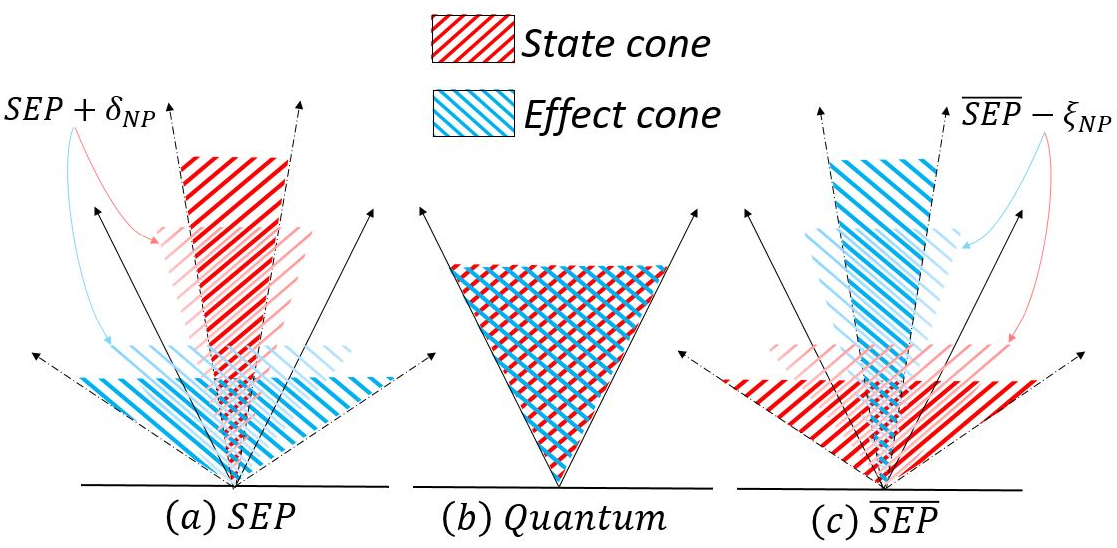}
\caption{(Color online) Intuitive representation of the state and effect cones for different compositions of two individual quantum systems. (a) $SEP$ theory: only separable states are allowed, whereas effects are more exotic than the effects allowed in quantum theory; (b) Quantum theory: state cone and effect cone exactly overlap ({\it self dual}); (c) $\overline{SEP}$ theory: effect cone is constrained in comparison to quantum theory, but allows states that are more exotic than quantum states. State cone of $SEP+\delta_{NP}$ lies strictly in between $SEP$ and $\mathcal{Q}$, whereas for $\overline{SEP}-\xi_{NP}$ theory it lies strictly in between $\mathcal{Q}$ and $\overline{SEP}$. Examples of such compositions are shown in (a) \& (b) with light shaded cones.}\label{fig1}
\end{figure}

A set of states $\{\omega_i\}\subset\Omega$ are called perfectly distinguishable if there exists a measurement $M\equiv\{e_i~|~\sum_ie_i=u\}$ such that $e_i(\omega_j)=\delta_{ij}$. Given a system $\mathcal{S}\equiv(\Omega,E)$, the maximal cardinality of the set of states that can be perfectly distinguished is known as the operational dimension of the system. On the other hand, the maximal cardinality of the set of states that can be perfectly distinguished pairwise in known as the information dimension of the system \cite{Brunner14(1)}. Note that, in case of operational dimension only one measurement is allowed to distinguished the states, whereas for information dimension different measurements for distinguishing different pairs are allowed. 

Given two systems $\mathcal{S}_A\equiv(\Omega_A,E_A)$ and $\mathcal{S}_B\equiv(\Omega_B,E_B)$ the state space of the composite system $\mathcal{S}_{AB}\equiv(\Omega_{AB},E_{AB})$ is embedded in the tensor product space $V_A\otimes V_B$ although the choice of $\Omega_{AB}$ in not unique. However, the {\it no signaling} principle and {\it tomographic locality} \cite{Hardy11} postulate narrow down the choices within two extremes -- the minimal tensor product space and maximal tensor product space, resulting in theories we will call $SEP$ and $\overline{SEP}$, respectively. For two systems, each described by quantum theory individually, the state cone in $SEP$ theory is given by
\footnotesize
\begin{align*}
\left(V_{AB}^{SEP}\right)_+:=\left\{\sum_i\pi_i^A\otimes \pi_i^B~|~\forall~i,~\pi_i^A\in\mathcal{T}_+(\mathcal{H_A})~\&~\pi_i^B\in\mathcal{T}_+(\mathcal{H_B})\right\}.
\end{align*}    
\normalsize
The effect cone $E_{AB}^{SEP}$ is dual to $\Omega_{AB}^{SEP}$ and is given by
\footnotesize
\begin{align*}
\left(V_{AB}^{SEP}\right)_+^\star:=\left\{Y\in\mathcal{T}(\mathcal{H_A}\otimes\mathcal{H}_B)~|~\Tr(XY)\ge0~\forall~X\in\left(V_{AB}^{SEP}\right)_+\right\}.
\end{align*}    
\normalsize
$SEP$ theory only allows separable states whereas the state cone $(V_{AB}^{\mathcal{Q}})_+:=\left\{Y\in\mathcal{T}_+(\mathcal{H_A}\otimes\mathcal{H}_B)\right\}$ of quantum theory contains both product and entangled states. On the other hand, effects in $SEP$ theory can be more exotic than quantum entangled effects. For instance, entanglement witnesses \cite{Horodecki09} can be a valid effect in this theory as they yield positive probability on any separable state. Thus we have the following set inclusion relations
\footnotesize
\begin{align*}
\left(V_{AB}^{SEP}\right)_+\subset\left(V_{AB}^{\mathcal{Q}}\right)_+=\left(V_{AB}^{\mathcal{Q}}\right)_+^\star\subset\left(V_{AB}^{SEP}\right)_+^\star.
\end{align*}    
\normalsize
\begin{figure}[t!]
\centering
\includegraphics[width=0.45\textwidth]{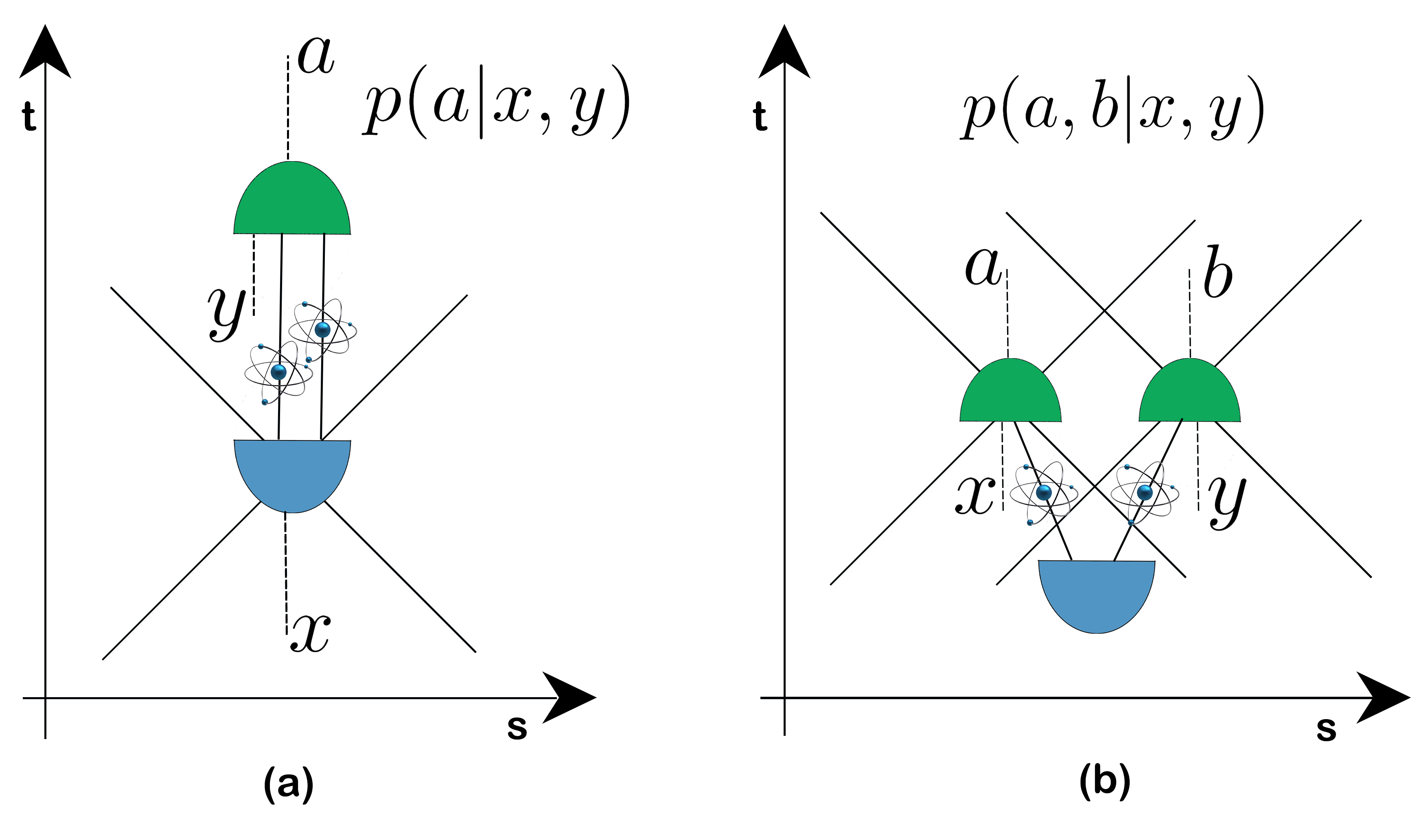}
\caption{(Color online) Time-like and space-like correlations. (a) the preparation (blue) and measurement (green) devices receive inputs $x$ and $y$, respectively and finally an outcome $a$ is obtained. The blue device prepares two elementary systems which are individually described as quantum bits, but the global description of the preparation is unknown. This setting can generate stronger-than-quantum correlation as shown in Theorem \ref{theo1}. Whereas in figure (b) correlations generated by space-like separated measurement devices acting on two qubits always imply a quantum description of the joint preparation\cite{Barnum10}. }\label{fig2}
\end{figure}
Equality between $\left(V_{AB}^{\mathcal{Q}}\right)_+~\&~\left(V_{AB}^{\mathcal{Q}}\right)_+^\star$ for quantum composition is known as {\it self duality}, which also holds true for a class of elementary GPT models \cite{Janotta11,Muller12,Banik19}. The $\overline{SEP}$ theory is obtained by interchanging the role of state cone and effect cone of $SEP$ theory (see Fig. \ref{fig1}) and hence allows more exotic states than quantum states, {\it e.g.} the {\it positive on pure tensors} states \cite{Klay87}.
In between $SEP$ and $\overline{SEP}$ one can consider many other compositions, either by adding a subset of non-separable states $\delta_{NP}$ with the state cone of $SEP$ or by excluding a subset of non-separable states $\xi_{NP}$ from the state cone of $\overline{SEP}$. The resulting theories we will denote as $SEP+\delta_{NP}$ and $\overline{SEP}-\xi_{NP}$. Excluding $SEP$, many of these theories are at par with quantum theory in the sense of containing nonlocal correlations. The result of Ref.  \cite{Barnum10} implies that one cannot obtain any space-like correlation in any of these models that is not available in quantum theory. More precisely, any input-output space-like correlation obtained in any of the above compositions between two qudits is, indeed, achievable in $\mathbb{C}^d\otimes_{\mathcal{Q}}\mathbb{C}^d$. In the following, we will show that such a thesis is not true anymore if we consider correlations in the time-like domain. Formally, space-like correlations represent joint input-output conditional probability distributions arising from local measurements performed on spatially separated composite systems where no communication is allowed between the subsystems. On the other hand, time-like scenario allows  communication from one subsystem to the other (see Fig.\ref{fig2}). While the  study of Bell nonlocality has motivated a vast literature in the former scenario \cite{Brunner14}, the latter avenue has been comparatively less explored. In the following, we introduce a game that helps to grasp the strength of correlations in the time-like domain.

{\it Pairwise distinguishability game $\mathcal{P}_D^{[n]}$.--} The game involves two players (say) Alice and Bob and a Referee. In each run of the game, the Referee provides a classical message $\eta$ to Alice, randomly chosen from some set of messages $\mathcal{N}$, where $|\mathcal{N}|:=n$; and asks Bob a question $\mathbb{Q}(\eta,\eta^\prime)$ -whether the message given to Alice in that run is $\eta\in\mathcal{N}$ or $\eta^\prime\in\mathcal{N}$, where $\eta^\prime\neq\eta$. Since $\eta^\prime\neq \eta$, $n \choose 2$ number of such different questions are possible. The winning condition of the game demands Bob answer all the questions correctly. Alice and Bob do not share any correlated state, but Alice can encode her message on the states of some physical system and accordingly send them to Bob. The following results provide a necessary and sufficient condition for winning the game in any GPT. 
\begin{proposition}\label{prop1}
Perfect winning of the game $\mathcal{P}_D^{[n]}$ requires Alice to encode her message $\eta\in\mathcal{N}$ on a set of states $\{\omega_\eta\}_{\eta\in\mathcal{N}}\subset\Omega$ of some system $\mathcal{S}\equiv(\Omega,E)$ such that the states within the set $\{\omega_\eta\}_{\eta\in\mathcal{N}}$ are pairwise distinguishable.  
\end{proposition}
Formal proof of the proposition we defer to supplemental \cite{Supple}. 
Here, we would like to point out that the concept of information dimension of the system used by Alice to encode her messages plays a crucial role in this game since in each run, the question $\mathbb{Q}(\eta,\eta^\prime)$ asked to Bob will be a function of two messages one of which has been provided to Alice. 

We will now consider the situation where Alice encodes her messages on the states of multiple qubits available to her. However, depending upon the composite structures assumed to model these multiple qubits, different number of qubits may be required to win the same game, which leads us to one of our core results. 
\begin{theorem}\label{theo1}
Four qubits communication from Alice to Bob is required for winning the game $\mathcal{P}_D^{[12]}$ when quantum composition is considered among the elementary systems, whereas two $SEP$-bits ({\it i.e.} two qubits in $SEP$ composition) suffice for winning this game. 
\end{theorem}    
\begin{proof}
(outline) For a quantum system, the information dimension is the same as its operational dimension, which is again the same as the dimension of the associated Hilbert space \cite{Brunner14(1)}. Since the Hilbert space dimension of $(\mathbb{C}^2)^{\otimes3}$ is $8$, according to Proposition \ref{prop1}, three qubits communication is not sufficient for winning the game $\mathcal{P}_D^{[12]}$ perfectly. However, four qubits communication suffices as the number of distinguishable (as well as pairwise distinguishable) states, in this case, are $16$. If we consider the $SEP$ composition between two qubits, then the following $12$ states $\mathcal{A}:=\{\ket{\kappa\kappa},\ket{\kappa\bar{\kappa}},\ket{\bar{\kappa}\kappa},\ket{\bar{\kappa}\bar{\kappa}}\}_{\kappa\in\{x,y,z\}}$ turn out to be pairwise distinguishable; where $\ket{\alpha\beta}:=\ket{\alpha}\otimes\ket{\beta}$ and $\ket{\kappa}~(\ket{\bar{\kappa}})$ is the eigenstate of Pauli operator $\sigma_\kappa$ with eigenvalue $+1~(-1)$, where $\kappa\in\{x,y,z\}$. While some pair of states, such as $\{\ket{xx},\ket{x\bar{x}}\}$, are perfectly distinguishable in quantum theory due to mutual orthogonality, some pairs, such as $\{\ket{xx},\ket{zz}\}$, consisting non-orthogonal states cannot be perfectly distinguished in quantum theory. However, as shown by Arai {\it et al.} such states can be perfectly distinguished in $SEP$ theory \cite{Arai19}. Complete analysis (along with the measurement) of pairwise distinguishability of the states $\mathcal{A}$ in $SEP$ theory is presented in supplemental \cite{Supple}. The game $\mathcal{P}_D^{[12]}$, thus, can be perfectly won in SEP theory if Alice encodes her message on these states. This completes the proof of our claim.
\end{proof}    
Theorem \ref{theo1} establishes communication advantage of $SEP$ composition over the other two compositions $\mathcal{Q}$ and $\overline{SEP}$, even though it is a local theory by construction. This theorem also implies a curious feature of the SEP theory, known as `dimension mismatch' - difference between measurement dimension and information dimension \cite{Brunner14(1)}.
\begin{corollary}\label{coro1}
SEP composition exhibits the phenomenon of dimension mismatch.
\end{corollary}
\begin{proof}
Measurement dimension of two SEP-bits is $4$, which follows from Proposition 2.5 of Ref.\cite{Arai19}, whereas our Theorem \ref{theo1} establishes its information dimension to be strictly greater than $4$, and completes the proof.
\end{proof}
Dimension mismatch and consequently the presence of stronger than quantum time-like correlation in the above result arises strictly from the choice of composition. One might ask whether the advantage of two qubits in $SEP$ composition for playing the $\mathcal{P}_D^{[n]}$ game can be made arbitrarily large. Our next result is a no-go answer to this question (proof provided in supplemental \cite{Supple}).   
\begin{lemma}\label{lemma1}
The game $\mathcal{P}_D^{[n]}$ cannot be won perfectly by communicating the encoded states chosen from the composite system $\mathbb{C}^2{\otimes}_{\min}\mathbb{C}^2$ whenever $n>12$. 
\end{lemma} 
However the advantage of $SEP$ composition over quantum theory can be increased if we start with more number of $SEP$-bits initially. 
\begin{theorem}\label{theo2}
$2k$ number of $SEP$-bits are sufficient for winning the game $\mathcal{P}_D^{[12^k]}$ perfectly, whereas it requires $2k+\lceil{k\log_23}\rceil$ number of qubits, with $k\in\mathbb{Z}_+$. 
\end{theorem}
Proof is provided in supplemental \cite{Supple}. While deriving this theorem we use the fact that Bob addresses at-most two $Sep$-bits together while decoding Alice's message. The gap between the number of $SEP$-bits and qubits might increase further if Bob addresses all the $SEP$-bits together. However, we leave this question open for future research. We rather move to a possible experimental implication of our study.

Novel experimental proposals bring adequate physical reasoning to the `mathematical fiction' of Hilbert space formulation of quantum theory. For instance, experimentally observed algebraic relationship among the coherent cross sections of scattering amplitudes in triple-slit experiment constitutes a test for complex versus quaternion quantum theory \cite{Peres79}. The experiment by Sinha {\it et al.} that rules out multi-order interference in quantum mechanics is worth mentioning at this point \cite{Sinha10}. In a similar spirit, a pertinent question to ask is which particular composite structure between two elementary qubits must be preferred \cite{Self1,DiVincenzo00}. At this point, one might wish to postulate a particular composite structure. Schr\"{o}dinger, for instance, found Quantum composition `rather discomforting' \cite{Schrodinger35} due to the peculiarity of quantum entanglement as demonstrated in EPR gedankenexperiment \cite{Einstein35}. The $SEP$ composition, which does not contain these 'discomforting' features, is immediately ruled out due to the seminal experiment by Aspect \& collaborators \cite{Aspect81} and the recent {\it loophole} free Bell tests \cite{Hensen15,BBT18,Rauch18} which validate the presence of nonlocal correlations in nature. There are still a number of different bipartite compositions, such as $SEP+\delta_{EP},~\overline{SEP}-\xi_{EP},~\overline{SEP}$, that, like the quantum composition $\mathcal{Q}$, incorporate nonlocal correlations and hence cannot be excluded from the Bell test's results. Furthermore, no such model can contain any space-like correlation which are not available in quantum theory \cite{Barnum10}. At this point our $\mathcal{P}_D^{[12]}$ game starts playing a crucial role. Perfect success of this game with communication of less than four qubits assures the presence of beyond quantum time-like correlation which indicates a departure from the composition rule adopted in quantum theory. In fact, our next result proposes a generic test in this direction. 
\begin{proposition}\label{prop2}
For $n>2^k$, the game $\mathcal{P}_D^{[n]}$ cannot be won with $k$ qubits communication from Alice to Bob if the composition rule is quantum.  
\end{proposition}
The proof simply follows from Proposition \ref{prop1} and the fact that the information dimension of a quantum system is the same as its Hilbert space dimension. A non-null result in Proposition \ref{prop2}, {\it i.e.} successful completion of the task $\mathcal{P}_D^{[>2^k]}$ with $k$ qubits communication, will indicate a departure from quantum composition rule, whereas null result builds confidence towards quantum composition.

{\it Discussions.--} The present letter initiates a novel paradigm towards experimental tests for derivation of the composition rule from quantum mechanics. Importantly, unlike the seminal Bell tests that deal with space-like correlations, our proposal is based on time-like correlations. In this regard, the recent works in \cite{DallArno17,Czekaj17} are worth mentioning. In particular, the authors in \cite{DallArno17} provide examples of stronger than quantum time-like correlations while considering two elementary systems. However, the elementary systems considered there are post-quantum in nature (particularly the {\it square bits}), which itself can generate stronger time-like correlation than a qubit as established through our $\mathcal{P}_D^{[n]}$ task. The possibility of stronger than quantum time-like correlations in this present work, therefore, emerges strictly form the choices of composition between the qubits. From a technical point of view, the authors in \cite{DallArno17} utilize a concept called `signaling dimension' which is motivated from the study made in \cite{Frenkel15}, whereas in our case, the concept of information dimension plays a crucial role. On the other hand, the approach in \cite{Czekaj17} involves calculating entropic quantities, which requires particular structure in a theory to be well defined \cite{Krumm17,Takakura19}. Our approach, however, involves a very intuitive notion of pairwise distinguishability. A more elaborate discussion regarding novelty of our method in comparison to the existing approaches is deferred to supplemental material.   

Our study also motivates a number of questions that might be interesting for further exploration. First of all, dimension mismatch studied in \cite{Brunner14(1)} for {\it squre bit} model suggests several exotic implications. For instance, it can result in collapse of communication complexity and can also empower the Maxwell's demon indicating a violation of the second law of thermodynamics. Similar studies will be interesting in our case as our Corollary \ref{coro1} implies a dimension mismatch arising strictly from a compositional aspect of local quantum systems. It will also be interesting from a complexity theory perspective to answer the question mentioned after Theorem \ref{theo2}. Finding the implications of different composition rules are worth exploring for higher dimensional elementary quantum systems. Similar questions might also be re-framed in field theoretic formalism \cite{Efimov,Fainberg}. On the other hand, studying the implication of such stronger time-like correlations within the framework of generalized probabilistic theory might also provide fundamentally new insights regarding the structure of spacetime. The toy models of polygonal theory that have been studied extensively in the recent past \cite{Janotta11,Muller12,Banik19,Bhattacharya20,Saha20,Saha20(1)} might be a starting point towards this endeavour.   

{\bf Acknowledgement:} RKP acknowledges support through PMRF scheme. TG was supported by the Hong Kong Research Grant Council through grant 17307719 and though the Senior Research Fellowship Scheme SRFS2021-7S02, by the Croucher Foundation,  and by the John Templeton Foundation through grant  61466, The Quantum Information Structure of Spacetime  (qiss.fr). The opinions expressed in this publication are those of the authors and do not necessarily reflect the views of the John Templeton Foundation. SSB acknowledges partial support by the Foundation for Polish Science (IRAP project, ICTQT, Contract No. MAB/2018/5, co-financed by EU within Smart Growth Operational Programme). MA and MB acknowledge support through the research grant of INSPIRE Faculty fellowship from the Department of Science and Technology, Government of India. MB acknowledges funding from the National Mission in Interdisciplinary Cyber-Physical systems from the Department of Science and Technology through the I-HUB Quantum Technology Foundation (Grant No. I-HUB/PDF/2021-22/008) and the start-up research grant from SERB, Department of Science and Technology (Grant No. SRG/2021/000267).

\onecolumngrid
\section{Supplemental}
\section{Pairwise distinguishability game $\mathcal{P}_D^{[n]}$}
\begin{figure}[h!]
\centering
\includegraphics[height=6cm,width=6cm]{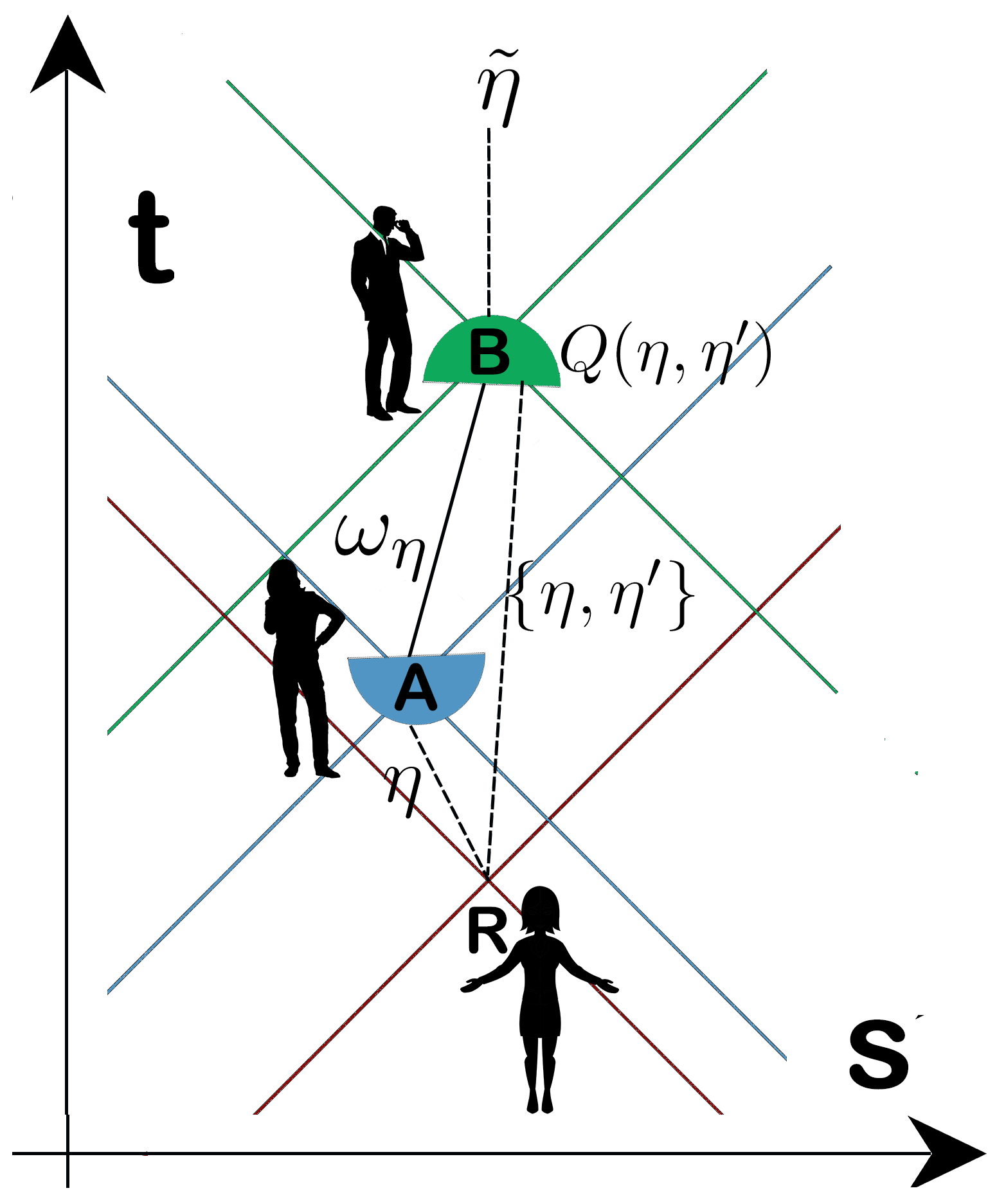}
\caption{(Color online) $\mathcal{P}_D^{[n]}$ game: in each run the Referee provides a randomly chosen message $\eta\in\mathcal{X}$ to Alice and accordingly asks Bob the question $\mathbb{Q}(\eta,\eta^\prime)$. Alice and Bob do not share any correlation, but Alice is allowed to communicate to Bob as she is in the past light cone of Bob (time is along vertical axis and horizontal axis represents space). Alice applies an encoding map $\mathcal{E}:\mathcal{X}\mapsto\{\omega_\eta\}_{\eta\in\mathcal{X}}\subset\Omega$, and sends the encoded state to Bob. Perfect winning of the game requires the states in $\{\omega_\eta\}_{\eta\in\mathcal{X}}$ to be pairwise distinguishable (Proposition 1).}\label{fig3}
\end{figure} 

\section{Proof of Proposition ${\bf 1}$}
\begin{proof}
In each run of the game $\mathcal{P}^{[n]}_D$, Referee randomly chooses a classical message $\eta$ from the set $\mathcal{N}$, where $|\mathcal{N}|:=n$; and sends it to Alice. Thereafter, Referee asks Bob the question $\mathbb{Q}(\eta,\eta^\prime)$ -- whether the message given to Alice in that run is $\eta\in\mathcal{N}$ or $\eta^\prime\in\mathcal{N}$, where $\eta^\prime\neq\eta$. However, if no communication is allowed, then Bob would choose randomly and in that case success probability would be $\frac{1}{2}$. But, if Alice is allowed to communicate then she would encode her received messages $\eta \in \mathcal{N}$ on a set of states $\{\omega_\eta\}_{\eta\in\mathcal{N}}\subset\Omega$ of some system $\mathcal{S}\equiv(\Omega,E)$. Afterward, the state is sent to Bob who already knows the encoding strategies. Upon receiving the question $\mathbb{Q}(\eta,\eta^\prime)$, he would try to distinguish the states $\omega_{\eta}$ and $\omega_{\eta^\prime}$ through some measurement. To win this game perfectly, Alice must choose pairwise distinguishable states $\{\omega_{\eta}\}_{\eta\in\mathcal{N}}$ to encode $n$ number of messages. For that the information dimension $\mathcal{I}$ of the system $\mathcal{S}$ must be lower bounded by $n$ i.e., $\mathcal{I}(\mathcal{S})\geq n$.  . \end{proof} 

\section{Detailed Proof of Theorem ${\bf 1}$}
\begin{proof}
In manuscript, we have already argued that four qubits are necessary to win $\mathcal{P}_D^{[12]}$ perfectly. Here we show that two SEP-bits suffice for winning this game. Alice starts by choosing the following encoding states  
\begin{align}
\Omega[12]:=\left\{\!\begin{aligned}
\omega_{zz}=\ket{zz}\bra{zz},~~~~\omega_{xx}=\ket{xx}\bra{xx},~~~~\omega_{yy}=\ket{yy}\bra{yy}\\
\omega_{z\bar{z}}=\ket{z\bar{z}}\bra{z\bar{z}},~~~~\omega_{x\bar{x}}=\ket{x\bar{x}}\bra{x\bar{x}},~~~~\omega_{y\bar{y}}=\ket{y\bar{y}}\bra{y\bar{y}}\\
\omega_{\bar{z}z}=\ket{\bar{z}z}\bra{\bar{z}z},~~~~\omega_{\bar{x}x}=\ket{\bar{x}x}\bra{\bar{x}x},~~~~\omega_{\bar{y}y}=\ket{\bar{y}y}\bra{\bar{y}y}\\
\omega_{\bar{z}\bar{z}}=\ket{\bar{z}\bar{z}}\bra{\bar{z}\bar{z}},~~~~\omega_{\bar{x}\bar{x}}=\ket{\bar{x}\bar{x}}\bra{\bar{x}\bar{x}},~~~~\omega_{\bar{y}\bar{y}}=\ket{\bar{y}\bar{y}}\bra{\bar{y}\bar{y}}
\end{aligned}\right\}\subset \left(V^{SEP}_{\mathbb{C}^2\otimes\mathbb{C}^2}\right)_+,	
\end{align}
where $\ket{\alpha\beta}:=\ket{\alpha}\otimes\ket{\beta}$ and $\ket{\kappa}~(\ket{\bar{\kappa}})$ is the eigenstate of Pauli operator $\sigma_\kappa$ with eigenvalue $+1~(-1)$, with $\kappa\in\{x,y,z\}$. If two states are orthogonal then they can be perfectly distinguished in quantum theory and hence can also be distinguished in $SEP$ theory since all the quantum effects are also allowed in $SEP$ theory. However, in the above set there are states that are not orthogonal to each other. For winning the game $\mathcal{P}^{[12]}_D$ in $SEP$ theory we need to show that those states are also perfectly distinguishable in $SEP$ theory. This immediately follows from the result of Arai {\it et al.} \cite{Arai19}. For completeness of our proof, here we argued the same and provide the detailed calculations. Consider the following two operators
\begin{align}
 E_1=\begin{pmatrix}
0 & 0 & 0 & ~~\frac{1}{2}\\
0 & 1 & ~~\frac{1}{2} & 0\\
0 & ~~\frac{1}{2} & 1 & 0\\
~~\frac{1}{2} & 0 & 0 & 0\\
\end{pmatrix}~~~~~\&~~~~~~~
E_2=\begin{pmatrix}
1 & 0 & 0 & -\frac{1}{2}\\
0 & 0 & -\frac{1}{2} & 0\\
0 & -\frac{1}{2} & 0 & 0\\
-\frac{1}{2} & 0 & 0 & 1\\
\end{pmatrix}.
\end{align}
Consider an arbitrary bi-partite product state $\omega_{nm}= \frac{1}{2}[\mathbb{I}+\vec{n}.\vec{\sigma}]\otimes \frac{1}{2}[\mathbb{I}+\vec{m}.\vec{\sigma}]$ where $\vec{n}\equiv(n_1,n_2,n_3)$ and $\vec{m}\equiv (m_1,m_2,m_3)$ are the Bloch vectors of the corresponding systems. Straightforward calculation yields 
\begin{align}\label{sepm}
 \tr(E_1\omega_{nm})=\frac{1+n_1m_1-n_3m_3}{2} ~~ \& ~~ \tr(E_2\omega_{nm})=\frac{1-n_1m_1+n_3m_3}{2}.
\end{align}
For arbitrary choices of Bloch vectors $\vec{n}~\&~\vec{m}$ the above expressions are always positive which assures that  $E_1,E_2\in (V^{SEP}_{\mathbb{C}^2\otimes\mathbb{C}^2})_+^{*}$. Furthermore the fact $E_1+E_2=\mathbf{1}$ ensures that $\mathcal{M}\equiv\{E_1,E_2\}$ is a valid measurement is $SEP$ theory. From the expression in Eq.(\ref{sepm}) it is immediate that
\begin{align*}
\tr(E_1\omega_{xx})=\tr(E_1\ket{xx}\bra{xx})=1,~~~~\tr(E_1\omega_{zz})=\tr(E_1\ket{zz}\bra{zz})=0,\\
\tr(E_2\omega_{xx})=\tr(E_2\ket{xx}\bra{xx})=0,~~~~\tr(E_2\omega_{zz})=\tr(E_2\ket{zz}\bra{zz})=1.
\end{align*}
Thus the pair of (nonorthogonal) states $\{\omega_{xx},\omega_{zz}\}$ can be distinguished perfectly in $SEP$ theory with the measurement $\mathcal{M}\equiv\{E_1,E_2\}$. 

We will now argue that similar measurements can be constructed for any pair of non-orthogonal states in the set $\Omega[12]$. For that we first note down the following observation.
\begin{observation}
For any $E\in(V^{SEP}_{AB})_+^{*}$ we have $(U_A\otimes U_B)E(U_A^\dagger\otimes U_B^\dagger)\in(V^{SEP}_{AB})_+^{*}$ where $U_A$ is a unitary on $\mathcal{H}_A$ and $U_B$ is a unitary on $\mathcal{H}_B$.
\end{observation}
Consider the following set of unitary operations acting on $\mathbb{C}^2$
\begin{align}
\mathcal{U}\equiv\left\{\!\begin{aligned}
A_0^y:=\frac{1}{\sqrt{2}}\begin{pmatrix}
    1 & -i\\
    -i & 1\\
    \end{pmatrix},~
A_0^z:=\begin{pmatrix}
     1 & 0\\
    0 & 1\\
    \end{pmatrix},~
B_0^x:=\begin{pmatrix}
     1 & 0\\
    0 & 1\\
\end{pmatrix},~
B_0^y:=\begin{pmatrix}
     1 & 0\\
    0 & -i\\
    \end{pmatrix}\\
A_1^y:=\frac{1}{\sqrt{2}}\begin{pmatrix}
    1 & i\\
    i & 1\\
    \end{pmatrix},~
A_1^z:=\begin{pmatrix}
     0 & -i\\
    i & 0\\
    \end{pmatrix},~
B_1^x:=\begin{pmatrix}
     1 & 0\\
    0 & -1\\
    \end{pmatrix},~
B_1^y:=\begin{pmatrix}
     1 & 0\\
    0 & i\\
    \end{pmatrix}
\end{aligned}\right\}.
\end{align}

\begin{table}[h!]
\begin{tabular}{ |p{1cm}||p{1.9cm}|p{1.9cm}|p{1.9cm}|p{1.9cm}||p{1.9cm}|p{1.9cm}|p{1.9cm}|p{1.9cm}| }
\hline
&\cellcolor{green!15}&\cellcolor{green!15} &\cellcolor{green!15}&\cellcolor{green!15}&\cellcolor{brown!25} &\cellcolor{brown!25}&\cellcolor{brown!25}&\cellcolor{brown!25}\\
&\cellcolor{green!15}~~~~~~~~$\omega_{yy}$ &\cellcolor{green!15}~~~~~~~~$\omega_{y\bar{y}}$ &\cellcolor{green!15}~~~~~~~~$\omega_{\bar{y}y}$&\cellcolor{green!15}~~~~~~~~$\omega_{\bar{y}\bar{y}}$&\cellcolor{brown!25}~~~~~~~~$\omega_{zz}$ &\cellcolor{brown!25}~~~~~~~~$\omega_{z\bar{z}}$ &\cellcolor{brown!25}~~~~~~~~$\omega_{\bar{z}z}$&\cellcolor{brown!25}~~~~~~~~$\omega_{\bar{z}\bar{z}}$\\
&\cellcolor{green!15}&\cellcolor{green!15} &\cellcolor{green!15}&\cellcolor{green!15}&\cellcolor{brown!25} &\cellcolor{brown!25}&\cellcolor{brown!25}&\cellcolor{brown!25}\\
\hline
\cellcolor{yellow!25}&&&&&&&&\\
\cellcolor{yellow!25}$~~\omega_{xx}$&$B_0^xA_0^y\otimes B_0^xA_0^y$  &$B_0^xA_0^y\otimes B_0^xA_1^y$&$B_0^xA_1^y\otimes B_0^xA_0^y$ &$B_0^xA_1^y\otimes B_0^xA_1^y$&$B_0^xA_0^z\otimes B_0^xA_0^z $&$B_0^xA_0^z\otimes B_1^xA_1^z$&$B_1^xA_1^z\otimes B_0^xA_0^z $&$B_1^xA_1^z\otimes B_1^xA_1^z$\\
\cellcolor{yellow!25}&&&&&&&&\\
\hline
\cellcolor{yellow!25}&&&&&&&&\\
\cellcolor{yellow!25}$~~\omega_{x\bar{x}}$&$B_0^xA_0^y\otimes B_1^xA_0^y$  &$B_0^xA_0^y\otimes B_1^xA_1^y$&$B_0^xA_1^y\otimes B_1^xA_0^y$ &$B_0^xA_1^y\otimes B_1^xA_1^y$&$B_0^xA_0^z\otimes B_1^xA_0^z $&$B_0^xA_0^z\otimes B_0^xA_1^z$&$B_1^xA_1^z\otimes B_1^xA_0^z $&$B_1^xA_1^z\otimes B_0^xA_1^z$\\
\cellcolor{yellow!25}&&&&&&&&\\
\hline
\cellcolor{yellow!25}&&&&&&&&\\
\cellcolor{yellow!25}$~~\omega_{\bar{x}x}$&$B_1^xA_0^y\otimes B_0^xA_0^y$  &$B_1^xA_0^y\otimes B_0^xA_1^y$&$B_1^xA_1^y\otimes B_0^xA_0^y$ &$B_1^xA_1^y\otimes B_0^xA_1^y$&$B_1^xA_0^z\otimes B_0^xA_0^z $&$B_1^xA_0^z\otimes B_1^xA_1^z$&$B_0^xA_1^z\otimes B_0^xA_0^z $&$B_0^xA_1^z\otimes B_1^xA_1^z$\\
\cellcolor{yellow!25}&&&&&&&&\\
\hline
\cellcolor{yellow!25}&&&&&&&&\\
\cellcolor{yellow!25}$~~\omega_{\bar{x}\bar{x}}$&$B_1^xA_0^y\otimes B_1^xA_0^y$  &$B_1^xA_0^y\otimes B_1^xA_1^y$&$B_1^xA_1^y\otimes B_1^xA_0^y$ &$B_1^xA_1^y\otimes B_1^xA_1^y$&$B_1^xA_0^z\otimes B_1^xA_0^z $&$B_1^xA_0^z\otimes B_0^xA_1^z$&$B_0^xA_1^z\otimes B_1^xA_0^z $&$B_0^xA_1^z\otimes B_0^xA_1^z$\\
\cellcolor{yellow!25}&&&&&&&&\\
\hline
\cellcolor{green!15}&\cellcolor{red!15}&\cellcolor{blue!15}&\cellcolor{blue!15}&\cellcolor{blue!15}&&&&\\
\cellcolor{green!15}$~~\omega_{yy}$&\cellcolor{red!15}$~~~~~~~NA$  &\cellcolor{blue!15}$~~~~~~~QD$&\cellcolor{blue!15}$~~~~~~~QD$ &\cellcolor{blue!15}$~~~~~~~QD$&$B_0^yA_0^z\otimes B_0^yA_0^z $&$B_0^yA_0^z\otimes B_0^yA_1^z$&$B_0^yA_1^z\otimes B_0^yA_0^z $&$B_0^yA_1^z\otimes B_0^yA_1^z$\\
\cellcolor{green!15}&\cellcolor{red!15}&\cellcolor{blue!15}&\cellcolor{blue!15}&\cellcolor{blue!15}&&&&\\
\hline
\cellcolor{green!15}&\cellcolor{blue!15}&\cellcolor{red!15}&\cellcolor{blue!15}&\cellcolor{blue!15}&&&&\\
\cellcolor{green!15}$~~\omega_{y\bar{y}}$&\cellcolor{blue!15}$~~~~~~~QD$  &\cellcolor{red!15}$~~~~~~~NA$&\cellcolor{blue!15}$~~~~~~~QD$ &\cellcolor{blue!15}$~~~~~~~QD$&$B_0^yA_0^z\otimes B_1^yA_0^z $&$B_0^yA_0^z\otimes B_1^yA_1^z$&$B_0^yA_1^z\otimes B_1^yA_0^z $&$B_0^yA_1^z\otimes B_1^yA_1^z$\\
\cellcolor{green!15}&\cellcolor{blue!15}&\cellcolor{red!15}&\cellcolor{blue!15}&\cellcolor{blue!15}&&&&\\
\hline
\cellcolor{green!15}&\cellcolor{blue!15}&\cellcolor{blue!15}&\cellcolor{red!15}&\cellcolor{blue!15}&&&&\\
\cellcolor{green!15}$~~\omega_{\bar{y}y}$&\cellcolor{blue!15}$~~~~~~~QD$  &\cellcolor{blue!15}$~~~~~~~QD$&\cellcolor{red!15}$~~~~~~~NA$ &\cellcolor{blue!15}$~~~~~~~QD$&$B_1^yA_0^z\otimes B_0^yA_0^z $&$B_1^yA_0^z\otimes B_0^yA_1^z$&$B_1^yA_1^z\otimes B_0^yA_0^z $&$B_1^yA_1^z\otimes B_0^yA_1^z$\\
\cellcolor{green!15}&\cellcolor{blue!15}&\cellcolor{blue!15}&\cellcolor{red!15}&\cellcolor{blue!15}&&&&\\
\hline
\cellcolor{green!15}&\cellcolor{blue!15}&\cellcolor{blue!15}&\cellcolor{blue!15}&\cellcolor{red!15}&&&&\\
\cellcolor{green!15}$~~\omega_{\bar{y}\bar{y}}$&\cellcolor{blue!15}$~~~~~~~QD$  &\cellcolor{blue!15}$~~~~~~~QD$&\cellcolor{blue!15}$~~~~~~~QD$ &\cellcolor{red!15}$~~~~~~~NA$&$B_1^yA_0^z\otimes B_1^yA_0^z $&$B_1^yA_0^z\otimes B_1^yA_1^z$&$B_1^yA_1^z\otimes B_1^yA_0^z $&$B_1^yA_1^z\otimes B_1^yA_1^z$\\
\cellcolor{green!15}&\cellcolor{blue!15}&\cellcolor{blue!15}&\cellcolor{blue!15}&\cellcolor{red!15}&&&&\\
\hline
\end{tabular}
\caption{Each cell here denotes the explicit form of unitary $U_1\otimes U_2$. For example ($B_0^xA_0^y\otimes B_0^xA_0^y$) can be used to construct the measurement $\mathcal{M}^\prime\equiv\left\{((B_0^xA_0^y)^\dagger\otimes (B_0^xA_0^y)^\dagger)E_1(B_0^xA_0^y\otimes B_0^xA_0^y),((B_0^xA_0^y)^\dagger\otimes (B_0^xA_0^y)^\dagger)E_2(B_0^xA_0^y\otimes B_0^xA_0^y)\right\}$ to distinguish $\ket{xx}$ and $\ket{yy}$ perfectly.) $QD$ stands for states which are distinguishable in Ordinary Quantum Composition. These states are definitely distinguishable in $SEP$ composition since effects allowed in Quantum theory are also allowed in $SEP$ theory. The cells named $NA$ are invalid questions as $\eta^\prime\neq\eta$ in $\mathbb{Q}(\eta,\eta^\prime)$.}
\label{tab1}
\end{table}
Any pair of non-orthogonal states in $\Omega[12]$ can be perfectly distinguished with a measurement $\mathcal{M}^\prime\equiv\{E_1^\prime,E_2^\prime\}$ which is connected to the measurement $\mathcal{M}$ through some product unitary $U_i\otimes U_j$, where $U_i,U_j\in\mathcal{U}$. The measurement $\mathcal{M}^\prime$ is therefore $\mathcal{M}^\prime\equiv\left\{(U_i^\dagger\otimes U_j^\dagger)E_1(U_i\otimes U_j),~~(U_i^\dagger\otimes U_j^\dagger)E_2(U_i\otimes U_j)\right\}$. Proper choices of the unitaries for different pairs are listed in Table \ref{tab1}.
\end{proof}
\section{Proof of Lemma ${\bf 1}$}
\begin{proof}
System $\mathcal{S}\equiv(\Omega,E)$, having the information dimension $\mathcal{I}(\mathcal{S})=k$ can always be used to encode the classical message to win the game $\mathcal{P}_D^{[n\leq k]}$ perfectly. As per definition, $\mathcal{I}(\mathcal{S})=k$ provides the upper bound on the number of pairwise distinguishable states. Here we will calculate the Information dimension of the most elementary SEP i.e., $\mathbb{C}^2\otimes_{\min}\mathbb{C}^2$, and show that $\mathcal{I}(\mathbb{C}^2\otimes_{\min}\mathbb{C}^2)=12$. So no $\mathcal{P}^{[n>12]}_D$ game can be won perfectly considering the above system. As shown by Arai {\it et.al.} \cite{Arai19}, two pure states $\rho_1=\rho_1^A\otimes\rho_1^B$ and $\rho_2=\rho_2^A\otimes\rho_2^B$ are perfectly distinguishable in SEP if and only if 
\begin{align}
\tr\rho_1^A\rho_2^A+\tr\rho_1^B\rho_2^B\le1.
\end{align}
Expressing $\rho_i^X=\frac{1}{2}\left(\mathbf{1}+\hat{n}_i^X\cdot\vec{\sigma}\right)$, where $\hat{n}_i^X\in\mathbb{R}^3$ with $|\hat{n}_i^X|=1$ for $i\in\{1,2\}$ and for $X\in\{A,B\}$, the above condition can be re-written as
\begin{align}
\hat{n}_1^A\cdot \hat{n_2}^A+\hat{n}_1^B\cdot \hat{n}_2^B\le0.
\end{align}
Let us define $N_i:= (\hat{n}_i^A,\hat{n}_i^B)^T \in\mathbb{R}^3\oplus\mathbb{R}^3\equiv\mathbb{R}^6$. Thus the above condition reads as 
\begin{align}
N^T_1.N_2 \leq 0.
\end{align}
Therefore, our question boils down to finding the maximum number of vectors in $\mathbb{R}^6$ such that the inner product between any two vectors is not positive definite. 
 
It is not hard to argue that in a  \textit{d}-dimensional vector space at the most $2d$ number of such vectors can be drawn. To appreciate the argument, observe that, in $1$-D only two such vectors can be drawn trivially, while in $2$-D the number is \textit{four}. Now, for $\mathbb{R}^{d}$ we can decompose it as $\mathbb{R}^{d-1}\oplus\mathbb{R}$. As mentioned earlier the $d^{th}-$dimension $\mathbb{R}$ contains only two such vectors and $\mathbb{R}^{d-1}$ can be further decomposed as $\mathbb{R}^{d-2}\oplus\mathbb{R}$. Thus, it can be argued from the method of induction that the maximum number of such vectors is $2d$.
\end{proof}

\section{Proof of Theorem ${\bf 2}$}
\begin{proof}
Let us denote the bipartite-elementary system in $SEP$ as $\mathcal{S}_E= \mathbb{}{C}^2\otimes_{\min}\mathbb{C}^2$ whose state space is constructed through minimal tensor product of the state space of two qubit systems. We have already seen that there are at most twelve pairwise distinguishable states in $\mathcal{S}_E$ i.e., $\mathcal{I}(\mathcal{S}_E)=12$. Let us denote these twelve states as $i\in\{1,\cdots,12\}$, where any two of them such as $\{i,j\}$ can be pairwise distinguished through some elementary $SEP$ measurement whenever $i\neq j~\&~ i,j\in\{1,\cdots,12\}$.
 
Consider $k$ number of bipartite-elementary systems $\mathcal{S}^{\otimes_{\min} k}_E\equiv(\mathbb{C}^2\otimes_{\min}\mathbb{C}^2)^{\otimes_{\min}k}$ and the $2k$-parties product state 
\begin{align}
X=i_1\otimes i_2\otimes\cdot \cdot\cdot \otimes i_t\otimes \cdot\cdot\cdot\otimes i_k, \end{align}
where $i_t \in \{1,\cdots,12\}~\forall~t\in\{1,\cdots,k\}$ and $t$ denotes the $t^{th}$ bipartite-elementary SEP. Consider another such $2k$-parties product state
\begin{align}
Y=j_1\otimes j_2\otimes\cdot \cdot\cdot \otimes j_t\otimes \cdot\cdot\cdot\otimes j_k,
\end{align}
where $j_t\in\{1,\cdots,12\}$. The state $X$ will be perfectly distinguishable to the state $Y$ whenever $i_t \neq j_t$ for some (at least one) $t\in\{1,\cdots,k\}$. For that particular $t^{th}$ bipartite-elementary system Bob will follow the pairwise distinguishability strategy as discussed in the proof of Theorem 1. Therefore, for $2k$ number of $SEP$-bits we can construct a set of $12^k$ number of different states that are pairwise distinguishable by following the aforesaid procedure. 

The construction above ensures that while playing the game $\mathcal{P}^{[12^k]}_D$ with elementary $SEP$ bits, communication of $2k$ $SEP$ bits from Alice to Bob suffices for a perfect winning strategy. Alice encodes her messages is the set of states constructed above and Bob decodes the messages based on the question given to him. It is not hard to see that $2k + \lceil {k\log3}\rceil$ number of qubits communication is necessary for winning this game. Here, $\lceil{x}\rceil$ denotes the ceiling function of $x$ and $\log$ is in base 2. In other words, the advantage of $SEP$ composition over quantum composition increases linearly with the increase in number of elementary systems. Here we would like to mention that the advantage we have reported might not be optimal. There is potential to get more advantage. Note that during the decoding step Bob addresses (at most) two elementary systems together in the strategy we have considered. There is a possibility that the number of pairwise distinguishable states will increase if Bob addresses all the systems together during his decoding process. 
\end{proof}

\section{Time-like paradigm: two different scenarios}
While studying the correlation in time-like paradigm, the following two broad scenarios can be considered. 

{\bf Scenario-I:} Suppose Alice and Bob are two parties isolated in their respective lab. A referee gives some classical random variable $x\in X$ to Alice while Bob has to return back some classical random variable $b\in B$ to the referee who then provides some payoff $\mathcal{P}_I:X\times B\to\mathbb{R}$ to them as a team. Alice is allowed to send some physical system (classical / quantum / more general system) to Bob. These different physical systems are assumed to obey some common condition. For instance, the number of perfectly distinguishable states must be bounded by some natural number $n\in \mathbb{N}_+$, called measurement dimension \cite{Brunner14(1)}. For $n=2$, a classical two level system (bit), respectively a quantum two level system (qubit), or a toy system with square shaped state space (square bit) satisfy the aforesaid measurement dimension criterion. Alice and Bob are allowed to decide on some shared strategies for which they are allowed to use classical shared randomness.

In such a scenario, the seminal no-go result by Holevo \cite{Holevo73} establishes that a qubit is no better than a classical bit when the payoff is measured through the mutual information $I(X:B)$ between input random variable $X$ and output random variable $B$. More recently, Frenkel and Weiner have strengthened this no-go implication by establishing that a qubit is no better than a classical bit no matter what payoff is considered to compare their communication utilities \cite{Frenkel15}. With the initials of the authors we call this H-FW scenario. The result of \cite{Frenkel15} further motivates the authors in \cite{DallArno17} to introduce the concept called signalling dimension.
\begin{definition}
[PRL 119, 020401] The signaling dimension of a system $S$, denoted by $\kappa(S)$, is defined as the smallest integer $d$ such that $\mathbb{P}^{m\to n}_S\subseteq \mathbb{P}^{m\to n}_{C_d}$, for all $m$ and $n$.
\end{definition}
Where $\mathbb{P}^{m\to n}_{C_d}$ is the set of all $m$-input ($|X|=m$) and $n$-output ($|B|=n$) conditional probability distributions that can be obtained by means of a $d$-level classical noiseless channel and shared random data, and similarly the set $\mathbb{P}^{m\to n}_S$ is defined. The 'No-Hypersignaling Principle' (NHP) introduced in \cite{DallArno17} deals with how the signaling dimension should behave under composition of different systems. More precisely, the principle will be satisfied if and only if signaling dimension of a composite system $\otimes_k S_k$ satisfies $\kappa(\otimes_k S_k)\le\Pi_k\kappa(S_k)$.
\begin{figure}
\centering
\includegraphics[scale=0.6]{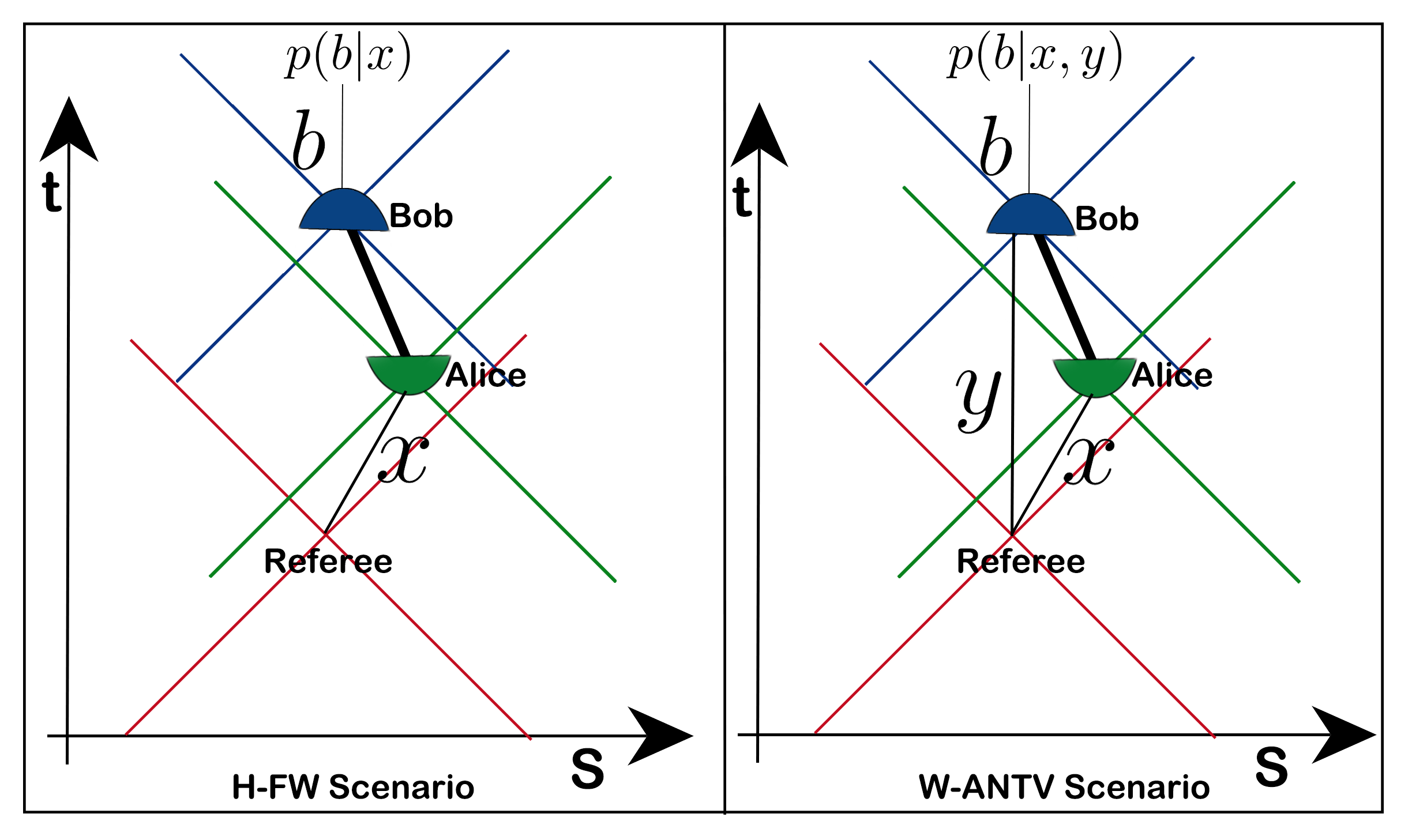}
\caption{(Color online) Two inequivalent scenarios in time-like paradigm. The figure on the left shows a schematic space-time diagram of the H-FW scenario \cite{Holevo73, Frenkel15, DallArno17} where Alice encodes a classical message $x$, received from the referee, into a system and sends it to Bob. Bob has to output a classical variable $b$ and they are rewarded based on the observed statistics $p(b|x)$ and a suitably chosen payoff. The figure on the right tells the story of the W-ANTV scenario \cite{Ambainis99,Ambainis02,Wiesner83,Czekaj17} where both Alice and Bob receive classical variables $x$ and $y$, respectively, from the referee. The input variables are not independently generated in general. Alice encodes her input in a system and sends it to Bob, who outputs the classical variable $b$. The reward in this scenario is based on the statistics $p(b|x,y)$ and a suitably chosen payoff.}    \label{fig:timelike}
\end{figure}

{\bf Scenario-II:} In this case, both Alice and Bob are given random variables $x\in X$ and $y\in Y$ (generally $y$ is some query regarding $x$) respectively and only Bob has to return some classical random variable $b\in B$ to the referee. Accordingly they are given some payoff $\mathcal{P}_{II}:X\times Y\times B\to \mathbb{R}$. The task of RAC studied in \cite{Ambainis99,Ambainis02} and originally proposed in \cite{Wiesner83} lies within this scenario, and hence we call it W-ANTV scenario.

Importantly, unlike the  Frenkel and Weiner work \cite{Frenkel15}, no such generic result is possible here. Depending on the payoff to quantify the communication utility, a qubit may or may not be advantageous over its classical counterpart. For instance, if we consider communication utility of a qubit system to be measured by its success in RAC game then it is strictly advantageous than a classical bit. On the other hand, if the utility is captured through the perfect success in the $\mathcal{P}_D^{[n]}$ game then a qubit is no good than a bit as its information dimension is identical to the measurement dimension. Interestingly, our results show that utility of two qubits in $\mathcal{P}_D^{[n]}$ depends on what kind of compositions are assumed between them. While quantum composition is not advantageous, the SEP composition turns out to be advantageous in this game and hence establishes stronger correlation in time-like paradigm. For a more detailed discussion of these two scenarios we refer to a recent work by the same group of authors \cite{Patra22}. 

Consider now the square bit model as studied in \cite{DallArno17}. The four pure states and four ray extremal effects are given by,
\begin{align*}
\omega_0=\begin{pmatrix}1\\0\\1 \end{pmatrix},~~ \omega_1=\begin{pmatrix}0\\1\\1 \end{pmatrix},~~
\omega_2=\begin{pmatrix}-1\\0\\1 \end{pmatrix},~~
\omega_3=\begin{pmatrix}0\\-1\\1 \end{pmatrix};\\
e_0=\begin{pmatrix}1\\1\\1 \end{pmatrix},~~
e_1=\begin{pmatrix}-1\\1\\1 \end{pmatrix},~~
e_2=\begin{pmatrix}-1\\-1\\1 \end{pmatrix},~~
e_3=\begin{pmatrix}1\\-1\\1 \end{pmatrix}.
\end{align*}
Any valid bipartite composition should allow the following $16$ product states and effects:
\begin{align*}
\Omega_{4i+j}:=\omega_i\otimes\omega_j^T,~~~~E_{4i+j}:=e_i\otimes e_j^T,~~~~~i,j\in\{0,1,2,3\}. 
\end{align*}
Furthermore there exist $8$ more states $\{\Omega_{16},\cdots,\Omega_{23}\}$ that are not product ({\it i.e.} entangled) but yields positive probability on any product effect. Similarly there are $8$ more entangled effects $\{E_{16},\cdots,E_{23}\}$ that yield consistent probability on any product state. Explicit form of these states and effects are available in \cite{DallArno17}. Note that square bit model has information dimension $4$ while its measurement dimension is $2$, and hence itself allows stronger time-like correlation than a qubit as established through the $\mathcal{P}_D^{[n]}$ game. However, the NHP approach cannot witness this stronger time-like correlation as NHP is applicable only to composite systems. 
The authors in \cite{DallArno17} have identified the following four consistent composite models: 
\begin{itemize}
\item {\bf PR model:} Contains all $24$ states and only product effects.
\item {\bf HS model:} Contains only product states and all $24$ effects.
\item {\bf Hybrid model:} Contains suitably chosen product and entangled states and effects.
\item {\bf Hybrid model:} Contains suitably chosen product and entangled states and effects allowing no dynamics.
\end{itemize}
and establish stronger than quantum time-like correlations for HS model and Hybrid model, while fails to do so for the PR model. Our approach however reveals even a stronger type of time-like correlation for PR model as all $24$ different states can be pairwise distinguished in this model leading to super-multiplicity of information dimension.

On the other hand, the $\mathcal{P}_D^{[n]}$ game as well as information content principle (ICP) studied in \cite{Czekaj17} fit within the W-ANTV scenario. Note that, ICP involves payoff quantified in terms of entropic quantity and uses the concept of strong subadditivity that holds for both the quantum and classical world. A careful analysis is required while studying the ICP in generalized probability theory framework. Here we note that if in some theory a mixed state allows non-unique decompositions in terms of different sets of perfectly distinguishable states with different probabilities, then a natural notion of entropy is not well defined in that theory \cite{Krumm17,Takakura19}. Now, consider the bipartite  state $\rho_{AB}=\frac{1}{2}\ket{00}_{AB}\bra{00}+\frac{1}{2}\ket{++}_{AB}\bra{++}$, whose spectral decomposition is given by 
\begin{align*}
\rho_{AB}&=\frac{3}{4}\ket{\psi_1}_{AB}\bra{\psi_1}+\frac{1}{4}\ket{\psi_2}_{AB}\bra{\psi_2}, ~~~\mbox{where},\\
\ket{\psi_1}&:=\frac{1}{\sqrt{12}}(3\ket{00}+\ket{01}+\ket{10}+\ket{11}),\\
\ket{\psi_2}&:=\frac{1}{2}(-\ket{00}+\ket{01}+\ket{10}+\ket{11})
\end{align*}
However, the pair of states $\{\ket{00},\ket{++}\}$ is perfectly distinguishable in the SEP theory. As a consequence the state $\rho$ can be assigned with two different entropies  $H\left(\frac{1}{2},\frac{1}{2},0,0\right)$ and $H\left(\frac{3}{4},\frac{1}{4},0,0\right)$, and fails to satisfy some consistency requirement; $H(q_i):=-\sum_iq_i\log q_i$. Thus the the study of ICP while considering different possible decompositions of local quantum system might require further assumptions. In that respect, $\mathcal{P}_D^{[n]}$ game is easy to analyse across these different compositions as the payoff in this case involves a very intuitive notion of pairwise distinguishability.

\section{Toy composition allowing exotic space-like and time-like correlations}
As discussed in the main manuscript, structure of a GPT is characterized by specifying three basic elements -- state space, effect space, and transformations, in particular reversible ones. Same is true when considering compositions. While considering quantum composition, the set of reversible transformations are the set of unitaries, product as well as the joint ones, and in this case all pure states as well as all ray extremal effects are connected to each other through some reversible transformation. However, for minimal and maximal composition one can consider only the product unitaries. For minimal composition it implies all states are connected but not all the ray extremal effects, while for maximal tensor product this is exactly opposite (see Table \ref{tab2}.)

{\bf Frozen model:} Let us now considered a model with allowed state $\ket{\psi^-}$ along with the separable state space. To make the theory consistent we consider the effect space to be $\left(V_{AB}^{SEP}\right)_+^\star~\setminus~\mathbb{E}_{\psi^-}$, where $\bra{\psi^-}E\ket{\psi^-}<0,~\forall~E\in\mathbb{E}_{\psi^-}\subset\left(V_{AB}^{SEP}\right)_+^\star$. Furthermore, we consider the model to be `frozen' by imposing the condition that no reversible dynamics is allowed. Toy example of such a model can be obtained by composing of two square bits \cite{DallArno17}. For this frozen model the state space lies strictly in between composite quantum state space and state space obtained by minimal composition. Since the set of effects $\mathbb{E}_{\psi^-}$ is discarded from the effect space, all the measurements listed in Table \ref{tab1} will not be allowed anymore. However, the set of product states $\Omega[5]:=\{\omega_{xx},\omega_{yy},\omega_{\Bar{y}\Bar{y}},\omega_{zz},\omega_{\Bar{z}\Bar{z}}\}$ is still perfectly pairwise distinguishable . Therefore, this frozen model allows time-like correlations stronger than quantum theory. This model will also allow nonlocal space-like correlation due to the presence of the entangled state $\ket{\psi^-}$.
\begin{table}[h!]
\centering
\begin{tabular}{|c||c|c|c|c|c|c|} 
\hline
~~\#~~&Composition&States($\Omega$) & Effects(E) & Reversible Transformation (RT) &\multicolumn{2}{c|}{Connectivity through RT} \\[0.6ex]\cline{6-7}
& & &  &  & ~~~~~States~~~~~ & Effects\\
\hline\hline
1&$\otimes_Q$~~~~&{$\left(V_{AB}^{Q}\right)$} & {$\left(V_{AB}^{Q}\right)$} &{$\{U_{AB},U_A\bigotimes U_B\}$}&\checkmark & \checkmark\\
\hline
2&$\otimes_{\min}$& {$\left(V_{AB}^{SEP}\right)_+$} & {$\left(V_{AB}^{SEP}\right)_+^\star$} &{$\{U_A\bigotimes U_B\}$}&\checkmark &$\times$ \\
\hline
3&$\otimes_{\max}$&{$\left(V_{AB}^{SEP}\right)_+^\star$} & {$\left(V_{AB}^{SEP}\right)_+$} &{$\{U_A\bigotimes U_B\}$}&$\times$&\checkmark\\
\hline
4&$\otimes_{\min}+NP$& {$\left(V_{AB}^{SEP}\right)_+~\bigcup~ \{\psi^-\}$} & {$\left(V_{AB}^{SEP}\right)_+^\star~\setminus~\mathbb{E}_{\psi^-}$} &{Frozen}&$\times$  &$\times$ \\
\hline
$\cdot$&$\cdot$&$\cdot$&$\cdot$&$\cdot$&$\cdot$&$\cdot$\\
\hline
\end{tabular}
\caption{Different compositions of two quantum systems. Quantum composition is very special due to self duality of the state \& effect cones, and all the pure states as well as the ray extremal effects are connected through RT. For minimal tensor product all pure state are connected through RT, but not all the ray extremal effects; and in the case of maximal tensor product the situation is exactly the opposite. From a dynamical point of view, the Frozen model is extremely constrained as it does not allow any RT. However, in this model one can have time-like correlations stronger than quantum theory, and this model also allows nonlocal space-like correlations.}
\label{tab2}
\end{table}

Naturally the question arises whether one can have a composition that will allow stronger than quantum correlations both in the time-like and space-like separated scenarios. Any state space lying in between quantum and minimal tensor product will not allow stronger than quantum correlations in the space-like domain because beyond-quantum effects cannot be realized in the space-like separated scenario. One must look for some state space lying strictly in between quantum and maximal tensor product. The result in \cite{Barnum10} implies that for such bipartite compositions one cannot obtain stronger than quantum correlations in the space-like domain in the standard Bell experiment. However, a recent result by some of the present authors shows that for such compositions one can have space-like correlations stronger than quantum when the Bell scenario is suitably generalized \cite{Lobo21}. Considering such a generalized scenario it is therefore possible to obtain stronger than quantum correlations both in time-like and space-like domains. We leave the explicit constructions as an open question for further research.


\begin{thebibliography}{99}
\bibitem{Self1} Recall that from a physical point of view, a qubit is simply a quantum two-level system, such as the two spin states of a spin $1/2$ particle, the ground and excited states of an atom, or the vertical and horizontal polarization of a single photon \cite{DiVincenzo00}.
\end{thebibliography}
\end{document}